%% file: qmonad.tex
\documentclass[11pt,a4paper]{article}
\usepackage{a4wide}
\usepackage{authblk}
\usepackage{bm}        % for math
\usepackage{dsfont}  
\usepackage{amssymb,amsmath,amsthm}   % for math
\usepackage{nicefrac}
\usepackage[PostScript=dvips]{diagrams}

\usepackage{tikz,ifdraft}
\usetikzlibrary{arrows,positioning,matrix,backgrounds,calc,fit} 

\newtheorem{theorem}             {Theorem}
\newtheorem{proposition}[theorem]{Proposition}
\newtheorem{lemma}      [theorem]{Lemma}

\newarrow{Eq}=====
\newarrow{mon}>--->
\newarrow{epi}----{>>}
\newarrow{mto}|--->
\newarrow{rel}--+->
\newarrow{inc}C--->
\newarrow{pfn}----{harpoon}
\diagramstyle[newarrowhead=vee,newarrowtail=vee]

\newcommand{\natarrow}{\rTo^{.}}

\newcommand{\Sig}{\sigma}
\newcommand{\sg}{\sigma}
\newcommand{\CS}{\mathcal{R}(\Sig)}
\newcommand{\CSf}{\mathcal{R}_{\fin}(\Sig)}
\newcommand{\Rs}{\CS}
\newcommand{\Rsf}{\CSf}
\newcommand{\fin}{f}

\newcommand{\id}{\mathsf{id}}

\newcommand{\ie}{i.e.~}

\newcommand{\rarr}{\rightarrow}

\newcommand{\va}{\mathbf{a}}

\newcommand{\vx}{\mathbf{x}}
\newcommand{\vy}{\mathbf{y}}

\newcommand{\vo}{\mathbf{o}}

\newcommand{\CC}{\mathcal{C}}

\newcommand{\ZZ}{\mathbb{Z}}
\newcommand{\MM}{\mathcal{M}}

\newcommand{\IMP}{ \; \Rightarrow \;}

\newcommand{\arq}{\stackrel{q}{\rightarrow}}
\newcommand{\arr}{\rightarrow}
\newcommand{\As}{\mathcal{A}}
\newcommand{\Bs}{\mathcal{B}}
\newcommand{\RA}[1]{R^{\mathcal{A}}_{#1}}
\newcommand{\RB}[1]{R^{\mathcal{B}}_{#1}}
\newcommand{\EE}{\mathcal{E}}
\newcommand{\EExiy}{\mathcal{E}^{i}_{\vx, y}}
\newcommand{\EEix}{\mathcal{E}^{i}_{\vx}}

\newcommand{\EExjy}{\mathcal{E}^{j}_{\vx', y}}
\newcommand{\EExy}{\mathcal{E}_{\vx, \vy}}
\newcommand{\EEx}{\mathcal{E}_{\vx}}
\newcommand{\EEax}{\mathcal{E}^{a}_{\vx}}
\newcommand{\EEaxy}{\mathcal{E}^{a}_{\vx, \vy}}

\newcommand{\FFx}{\mathcal{F}_{x}}
\newcommand{\FFxy}{\mathcal{F}_{x,y}}
\newcommand{\FFxyT}{\mathcal{F}_{x,y}^{T}}
\newcommand{\FFxypT}{\mathcal{F}_{x,y'}^{T}}

\newcommand{\HH}{\mathsf{H}}
\newcommand{\KK}{\mathsf{K}}
\newcommand{\Tr}{\mathsf{Tr}}
\newcommand{\Id}{\mathsf{Id}}

\newcommand{\Px}{P_{x}}
\newcommand{\Pxy}{P_{x,y}}
\newcommand{\Pxo}{P_{x,o}}
\newcommand{\Pvxvo}{P_{\vx,\vo}}
\newcommand{\Pxpop}{P_{x',o'}}
\newcommand{\Pxpyp}{P_{x',y'}}
\newcommand{\Qyz}{Q_{y,z}}
\newcommand{\Pxyp}{P_{x',y'}}
\newcommand{\Pxnyp}{P_{x,y'}}
\newcommand{\Pvx}{P_{\vx}}
\newcommand{\Pvxy}{P_{\vx,\vy}}

\newcommand{\vct}[1]{\mathsf{vec}( #1 )}
\newcommand{\Projd}{\mathsf{Proj}(d)}
\newcommand{\Projdp}{\mathsf{Proj}(d')}
\newcommand{\Projddp}{\mathsf{Proj}(dd')}
\newcommand{\Ze}{\mathbf{0}}
\newcommand{\Asq}{\sqrt{A}}
\newcommand{\Bsq}{\sqrt{B}}
\newcommand{\Alg}{\mathcal{A}}
\newcommand{\diag}{\mathsf{diag}}
\newcommand{\Complex}{\mathbb{C}}
\newcommand{\Qd}{\mathcal{Q}_{d}}
\newcommand{\Qdp}{\mathcal{Q}_{d'}}
\newcommand{\Qddp}{\mathcal{Q}_{dd'}}
\newcommand{\Qone}{\mathcal{Q}_{1}}
\newcommand{\Qa}{\mathcal{Q}_{a}}
\newcommand{\Qb}{\mathcal{Q}_{b}}
\newcommand{\Qc}{\mathcal{Q}_{c}}
\newcommand{\Qab}{\mathcal{Q}_{ab}}
\newcommand{\Qbc}{\mathcal{Q}_{bc}}
\newcommand{\Qabc}{\mathcal{Q}_{abc}}
\newcommand{\RQd}{R^{\mathcal{Q}_{d} \As}}
\newcommand{\RQdp}{R^{\mathcal{Q}_{d'} \As}}
\newcommand{\RQdQdp}{R^{\mathcal{Q}_{d}\mathcal{Q}_{d'} \As}}
\newcommand{\RQdB}{R^{\mathcal{Q}_{d} \Bs}}
\newcommand{\vp}{\mathbf{p}}
\newcommand{\Set}{\mathbf{Set}}
\newcommand{\unA}{\eta_{\As}}
\newcommand{\delx}{\delta_{x}}
\newcommand{\Qu}{\mathcal{Q}_{1}}
\newcommand{\mudd}{\mu^{d,d'}}
\newcommand{\muddA}{\mudd_{\As}}
\newcommand{\muddB}{\mudd_{\Bs}}
\newcommand{\Np}{\mathbb{N}^{+}}

\newcommand{\Cs}{\mathcal{\CC}}
\newcommand{\Kcomp}{\bullet}
\newcommand{\tddpAB}{m^{d,d'}_{\As, \Bs}}

\newcommand{\CSP}{\mathcal{K}}
\newcommand{\AK}{\As_{\CSP}}
\newcommand{\BK}{\Bs_{\CSP}}
\newcommand{\RcB}{R_{c}^{\BK}}
\newcommand{\RcA}{R_{c}^{\AK}}
\newcommand{\KAB}{\CSP_{AB}}
\newcommand{\ec}{e_{C}}
\newcommand{\Prob}{\mathsf{Prob}}
\newcommand{\Ke}{\CSP_{e}}
\newcommand{\AKe}{\As_{\Ke}}
\newcommand{\BKe}{\Bs_{\Ke}}

\newcommand{\bc}{b_{c}}
\newcommand{\pc}{p_{c}}
\newcommand{\Ax}{A_{x}}
\newcommand{\Axp}{A_{x'}}
\newcommand{\rxy}{r_{xy}}
\newcommand{\rxyp}{r_{x{y'}}}
\newcommand{\rxpyp}{r_{x'y'}}
\newcommand{\Qxyo}{Q_{xy,o}}
\newcommand{\Qxyf}{Q_{xy,0}}
\newcommand{\Qxyt}{Q_{xy,1}}
\newcommand{\Qxypt}{Q_{xy',1}}
\newcommand{\Qxpypt}{Q_{x'y',1}}

\begin{document}

\title{The Quantum Monad on Relational Structures}

\author[1]{Samson Abramsky}
\author[1]{Rui Soares Barbosa}
\author[2]{Nadish de Silva}
\author[2]{Octavio Zapata}
\affil[1]{Department of Computer Science, University of Oxford\\
\texttt{\{samson.abramsky,rui.soares.barbosa\}@cs.ox.ac.uk}}
\affil[2]{Department of Computer Science, University College London\\
\texttt{nadish.desilva@utoronto.ca, ocbzapata@gmail.com}}
\date{}
\maketitle

\begin{abstract}
Homomorphisms between relational structures play a central role in finite model theory, constraint satisfaction and database theory. 
%Much of graph theory can be formulated in terms of graph homomorphisms. 
%What could it mean to quantize these structures? 
A central theme in quantum computation is to show how quantum resources can be used to gain advantage in information processing tasks. 
In particular, non-local games have been used to exhibit quantum advantage in boolean constraint satisfaction, and to obtain quantum versions of graph invariants such as the chromatic number. 
%We use concepts from the semantics of computation to unify much of this work, and relate it to structural accounts of contextuality. 
We show how quantum strategies for homomorphism games between relational structures can be viewed as Kleisli morphisms for a quantum monad on the (classical) category of relational structures and homomorphisms. 
%We show a general connection between these notions and state-independent quantum realizations of strong contextuality in the Abramsky-Brandenburger formulation of contextuality. 
We use these results to exhibit a wide range of examples of contextuality-powered quantum advantage, and to unify several apparently diverse strands of previous work.
\end{abstract}

\section{Introduction}

Finite relational structures and the homomorphisms between them form a mathematical core common to finite model theory 
\cite{libkin2013elements}, constraint satisfaction \cite{feder1998computational}, and relational database theory \cite{kolaitis1998conjunctive}. Moreover, much of graph theory can be formulated in terms of the existence of graph homomorphisms, as expounded e.g.~in the influential text \cite{hell2004graphs}.
Thus, implicitly at least, the mathematical setting for all these works is categories of $\sigma$-structures and homomorphisms, for relational vocabularies $\Sig$.

What could it mean to quantize these structures? More precisely, with the advent of quantum computing, we can now consider the consequences of using quantum resources for carrying out various information-processing tasks. A major theme of current research is to delineate the scope of the \emph{quantum advantage} which can be gained by the use of quantum resources. How can this be related to these fundamental structures?

Our starting point is the notion of \emph{quantum graph homomorphism} introduced in \cite{manvcinska2016quantum} as a generalization of the notion of quantum chromatic number \cite{cameron2007quantum}.
Consider the following game, played by Alice and Bob cooperating against a Verifier. Their goal is to establish the existence of a homomorphism $G \rarr H$ for given graphs $G$ and $H$. Verifier provides vertices $v_1, v_2 \in V(G)$ to Alice and Bob respectively. They produce outputs $w_1, w_2 \in V(H)$ in response. No 
communication between Alice and Bob is permitted during the game. They win if the following conditions hold:
$v_1 = v_2 \IMP w_1 = w_2$ and $v_1 \sim v_2 \IMP w_1 \sim w_2$,
where we write $\sim$ for the adjacency relation.

If only classical resources are permitted, then the existence of a \emph{perfect strategy} for Alice and Bob --- one in which they win with probability 1 --- is equivalent to the existence of a graph homomorphism in the standard sense. However, using quantum resources, in the form of an entangled bipartite state where Alice and Bob can each perform measurements on their part, there are perfect strategies in cases where no classical homomorphism exists, thus exhibiting quantum advantage.

Alice--Bob games have also been studied for other tasks, notably for \emph{constraint systems}.
Consider the following system of linear equations over $\ZZ_2$:
\begin{center}
\begin{tabular}{lclcl}
$A \oplus B \oplus C = 0$    & $\quad$ &
$D \oplus E \oplus F = 0$ & $\quad$ & $G \oplus H \oplus I = 0$  \\
$A \oplus D \oplus G = 0$ & $\quad$ & $B \oplus E \oplus H = 0$ 
& $\quad$ & $C \oplus F \oplus I = 1$ \\
\end{tabular}
\end{center}
Of course, this system is not satisfiable in the standard sense, as we can see by summing over the left- and right-hand sides. Now consider the following Alice--Bob game. The Verifier sends Alice an equation, and Bob a variable. Alice returns an assignment to the variables in the equation, and Bob returns an assignment for his variable. They win if Bob's assignment agrees with Alice's, and moreover Alice's assignment satisfies the given equation. Classically, the existence of a perfect strategy is equivalent to the existence of a satisfying assignment for the whole system. Using quantum resources, there is a perfect strategy for the above system, which corresponds to Mermin's ``magic square'' construction \cite{Mermin90:SimpleUnifiedForm}. This can be generalized to a notion of quantum perfect strategies for a broad class of constraint systems \cite{cleve2014characterization,cleve2017perfect}, which have strong connections both to the study of contextuality in quantum mechanics, and to a number of challenging mathematical questions \cite{slofstra2016tsirelson,scholz2008tsirelson}.
Clearly, these games are analogous to those for graph homomorphisms. What is the precise relationship?

In \cite{manvcinska2016quantum}, generalizing results in \cite{cameron2007quantum}, the existence of a quantum perfect strategy for the homomorphism game from $G$ to $H$ is characterized in terms of the existence of a family  
$\{ E_{vw} \}_{v \in V(G), w \in V(H)}$ of projectors in $d$-dimensional Hilbert space for some $d$, subject to certain conditions. Analogous results for constraint systems are proved in \cite{cleve2014characterization}.
This characterization eliminates the two-person aspect of the game, and the shared state, leaving a ``projector-valued relation'' as the witness for existence of a quantum perfect strategy. We shall henceforth call these witnesses \emph{quantum graph homomorphisms}.
An important further step is taken in \cite{manvcinska2016quantum}. A construction $H \mapsto \mathsf{M} H$ on graphs is introduced, such that the existence of a quantum graph homomorphism from $G$ to $H$ is equivalent to the existence of a \emph{standard} graph homomorphism $G \arr  \mathsf{M} H$.

Our contribution begins at this point. We describe a general notion of non-local game for witnessing homomorphisms between structures for any relational signature. We show that the use of quantum resources in these games can be characterized by a notion of \emph{quantum homomorphism}. Moreover, quantum homomorphisms can in turn be characterized as the Kleisli morphisms for a \emph{quantum monad} on the (classical) category of relational structures and homomorphisms. This monad is \emph{graded}  \cite{milius2015generic} by the dimension of the Hilbert space.

Our account refines and generalizes the ideas from both \cite{cameron2007quantum,manvcinska2016quantum} and \cite{cleve2014characterization}. We characterize quantum solutions for general constraint satisfaction problems, showing as a special case that these subsume the binary constraint systems of \cite{cleve2014characterization}. We also show how quantum witnesses for strong contextuality in the sense of \cite{AbramskyBrandenburger} are characterized by quantum homomorphisms.

The precise relationship with the quantum graph homomorphisms of \cite{manvcinska2016quantum} turns out to be more subtle. We show that their notion is characterized by a quantum solution in our sense for a related boolean constraint system.
Overall, we show that a wide range of  notions of quantum advantage  is captured in a uniform way by the quantum monad, applied directly to the standard classical structures.

For reasons of limited space, some background material on linear algebra and quantum mechanics and some proofs have been relegated to an Appendix.

\section{From quantum perfect strategies to quantum homomorphisms}
\label{qpstoqhsec}

We write $[p] := \{ 1, \ldots , p \}$.
We fix a finite relational vocabulary $\sigma = \{ R_1, \ldots , R_p \}$, where $R_a$ has arity $k_a$, $a \in [p]$. 
A $\sg$-structure has the form $\As = (A, \RA{1}, \ldots , \RA{p})$, where $A$ is a non-empty set, and $\RA{a} \subseteq A^{k_{a}}$, $a \in [p]$. 
A homomorphism of $\sg$-structures $f : \As \rarr \Bs$ is a function $f : A \rarr B$ such that, for all $a \in [p]$ and $\vx \in A^{k_a}$: $\vx \in \RA{a} \IMP f(\vx) \in \RB{a}$.
Here we use vector notation: $\vx = (x_1, \ldots , x_{k_a})$, and $f(\vx) = (f(x_1), \ldots , f(x_{k_a}))$.
We denote the category of $\sg$-structures and homomorphisms by $\Rs$, and the full subcategory of finite structures by $\Rsf$.

We now consider the following game, played on finite structures $\As$, $\Bs$, in which Alice and Bob cooperate to convince a Verifier that there is a homomorphism from $\As$ to $\Bs$:
\begin{itemize}
\item Alice and Bob are separated, and not allowed to communicate (exchange classical information) while the game is played.
\item In a play of the game, the Verifier sends Alice an index $a$, and a tuple $\vx \in \RA{a}$; and  Bob an element $x \in A$.
\item Alice returns a tuple $\vy \in B^{k_a}$, and Bob returns an element $y \in B$.
\item Alice and Bob win that play if
\begin{itemize}
\item[(i)] $\vy \in \RB{a}$
\item[(ii)] $x = \vx_{i} \IMP y = \vy_{i}$, $i \in [k_{a}]$.
\end{itemize}
\end{itemize}
Alice and Bob may use probabilistic strategies. A \emph{perfect strategy} is one in which they win with probability~$1$.

It is clear that if only classical resources are allowed, the existence of a perfect strategy is equivalent to the existence of a homomorphism from $\As$ to $\Bs$. 
The actual strategy played by Alice and Bob may be pure or mixed, in the latter case using some shared randomness.
%In particular, note that condition (ii) implies that there is a function $f : A \rarr B$ such that Alice's response to $\vx$ is $f(\vx)$, and Bob's response to $x$ is $f(x)$; while condition (i) implies that this function is a homomorphism.\footnote{This function may be chosen using shared randomness.}

We now consider the use of quantum resources in the homomorphism game. We shall only consider the case of finite-dimensional resources in this paper. Such resources have the following general form:
\begin{itemize}
\item There are finite-dimensional Hilbert spaces $\HH$ and $\KK$, and a pure state $\psi$ on $\HH \otimes \KK$.
This state is shared between Alice and Bob. The separation between Alice and Bob is reflected in the fact that Alice can only perform operations on $\HH$, while Bob can only perform operations on $\KK$.
\item For each $a \in [p]$ and tuple $\vx \in \RA{a}$, Alice has a POVM $\EEax = \{ \EEaxy \}_{\vy \in B^{k_a}}$.
\item For each $x \in A$, Bob has a POVM $\FFx = \{ \FFxy \}_{y \in B}$
\end{itemize}
These resources are used as follows:
\begin{itemize}
\item Given $a$ and $\vx$, Alice measures $\EEax$ on her part of $\psi$.
\item Given $x$, Bob measures $\FFx$ on his part of $\psi$.
\item They obtain the joint outcome $(\vy, y)$ with probability $\psi^*(\EEaxy \otimes \FFxy)\psi$.
\end{itemize}
If with probability~1 the outcome $(\vy, y)$ satisfies the winning conditions, then this is a quantum perfect strategy.

We can write the winning conditions explicitly in terms of the quantum operations:
\[ \begin{array}{lll}
\mbox{(QS1)} & \psi^*(\EEaxy \otimes \FFxy)\psi = 0 & \quad \mbox{if $x = \vx_i$ and $y \neq \vy_i$} \\
\mbox{(QS2)} & \psi^*(\EEaxy \otimes I)\psi = 0 & \quad \mbox{if $\vy \not\in \RB{a}$} .
\end{array}
\]

A first remark is that our assumption of the bipartite structure of the state space does not in fact lose any generality. We could have asked simply that Bob's operators commute with those of Alice. However, since we are considering the finite-dimensional case, a result of Tsirelson \cite{tsirelson2006bell,scholz2008tsirelson} implies that this is equivalent to the tensor product formulation we have used.

Furthermore, using a pure state also does not lose any generality. Indeed, if we had a mixed state $\rho = \sum_i p_i \psi_i \psi_i^*$, with the trace replacing the inner products in (QS1) and (QS2), then the linearity of the trace implies that we could just as well use any of the pure states $\psi_i$ with the same measurements, and still satisfy the conditions.

We shall now show that in fact a quantum perfect strategy can without loss of generality be assumed to have a very special form, which will lead us to the equivalence with quantum homomorphisms.
These results combine ingredients from \cite{cleve2014characterization} and \cite{manvcinska2016quantum,cameron2007quantum}. The proofs are closest to those in \cite{roberson2013variations}, but are considerably simpler as well as more general.

For notational convenience, we shall focus on the case where the relational signature has a single $k$-ary relation $R$. Thus a quantum perfect strategy for the homomorphism game from $\As$ to $\Bs$ has the form 
$(\psi, \{ \EEx \}_{\vx \in \RA{}}, \{ \FFx \}_{x \in A})$,
where $\EEx = \{ \EExy \}_{\vy \in B^k}$ and $\FFx = \{ \FFxy \}_{y \in B}$ are POVMs satisfying the conditions (QS1) and (QS2).

Our first step is to show that without loss of generality, $\psi$ can be taken to have full Schmidt rank.
\begin{lemma}
\label{fullranklemm}
Given a quantum perfect strategy $(\psi', \{ \EEx' \}_{\vx \in \RA{}}, \{ \FFx' \}_{x \in A})$, we can find a strategy 
$(\psi, \{ \EEx \}_{\vx \in \RA{}}, \{ \FFx \}_{x \in A})$ where $\psi \in \Complex^d \otimes \Complex^d$ has the form $\sum_{i=1}^d \lambda_i e_i \otimes e_i$ with $\lambda_i > 0$ for all $i$.
\end{lemma}
\begin{proof}
Our proof will follow closely the first part of the proof of Theorem~6.5.1 in \cite{roberson2013variations}, so we omit detailed calculations.
We write the Schmidt decomposition of $\psi'$ as $\sum_{i=1}^d \lambda_i \alpha_i \otimes \beta_i \in \Complex^{d_A} \otimes \Complex^{d_B}$, where $\{ \alpha_i \}$, $\{ \beta_i \}$ are orthonormal families of vectors, and $\lambda_i > 0$.

We define $P_A := \sum_{i=1}^d e_i \alpha_i^*$, $P_B := \sum_{i=1}^d e_i \beta_i^*$. Thus $P_A : \Complex^{d_A} \rarr \Complex^{d}$, and $P_B : \Complex^{d_B} \rarr \Complex^{d}$. It is straightforward to verify that $P_A P_A^* = I_d = P_B P_B^*$. We have $\psi := (P_A \otimes P_B)\psi'$, and $\psi' = (P_A^* \otimes P_B^*)\psi$.

Similarly, we define $\EExy := P_A \EExy' P_A^*$, and $\FFxy := P_A \FFxy' P_B^*$. Again, it is straightforward to verify that this yields well-defined POVMs, and moreover that the probabilities are preserved:
$\psi^* (\EExy \otimes \FFxy) \psi = \psi'^* (\EExy' \otimes \FFxy') \psi'$.
Thus $(\psi, \{ \EEx \}_{\vx \in \RA{}}, \{ \FFx \}_{x \in A})$ is a quantum perfect strategy.
\end{proof}

The following simple general result will be useful.
\begin{lemma}
\label{idemlemm}
Let $\Alg$ be a $*$-algebra, and $a$, $b$, $d$ be self-adjoint elements of  $\Alg$, where $d$ is also  invertible.
Suppose that $ad = adb = db$. Then $a$ and $b$ are projectors, and they both commute with $d^2$.
If $\Alg$ is a $\mbox{C}^*$-algebra, then $a = b$.
\end{lemma}
\begin{proof}
First, $a^2d = adb = ad$. Since $d$ is invertible, this implies that $a^2 = a$, so $a$ is a projector. Similarly, $b^2 = b$.
Moreover, since $a$, $b$  and $d$ are self-adjoint, $db = ad \iff bd = da$.
Hence $ad^2 = dbd  = d^2a$.

If $\Alg$ is a $\mbox{C}^*$-algebra, then it is standard that $d$ commutes with every element which commutes with $d^2$ \cite{blackadar2006operator,pedersen2012analysis},
so $da = ad = db$, and since $d$ is invertible, this yields $a = b$.
\end{proof}

We shall now show that under the assumption of full Schmidt rank, the measurements are already remarkably constrained. 
We define $\EExiy := \sum_{\vy_i = y} \EExy$.
\begin{lemma}
\label{pvmtlemm}
Let $(\psi, \{ \EE_{\vx} \}, \{ \FFx \})$ be a quantum perfect strategy in which $\psi$ has full Schmidt rank.
Then for all $\vx$, $i$, $y$, $\EExiy$ and $\FFxy$ are projectors, and $\EExiy = \FFxyT$ whenever $x = \vx_i$.
\end{lemma}
\begin{proof}
We write $\psi$ as $\sum_{i=1}^d \lambda_i e_i \otimes e_i$, where $\lambda_i > 0$. The corresponding $d \times d$ diagonal matrix $D = \diag\{\lambda_i\}$ is full rank, and hence invertible, and $D^* = D$.
Note that $\psi = \vct{D}$, the vectorization of $D$. Using the standard equations 
$(A \otimes B)\vct{D} = \vct{ADB^{T}}$ and $\vct{A}^*\vct{B} = \Tr(AB)$, we have
\[ \psi^*(\EExy \otimes \FFxy)\psi = 0 \iff \Tr(D\EExy D\FFxyT) = 0 \iff \Tr(\EExy D\FFxyT D) = 0 . \]
By Proposition~\ref{trprop}, $\Tr(\EExy D\FFxyT D) = 0 \iff \EExy D\FFxyT D = \Ze$, and since $D$ is invertible, this is equivalent to $\EExy D\FFxyT = \Ze$. By the quantum strategy condition (QS1),
$\EExy D \FFxyT = \Ze$ when $x = \vx_i$ and $y \neq \vy_i$. This in turn implies that $\EExiy D \FFxypT = \Ze$
when $x = \vx_i$ and $y \neq y'$.

Now fix $\vx$ and $x = \vx_i$. Let $A_y := \EExiy$, and $B_y := \FFxyT$. We have $\sum_y A_y = I = \sum_y B_y$, and $A_y D B_{y'} = \Ze$ when $y \neq y'$. Hence
$A_y D \; = \; \sum_{y'} A_y D B_{y'} \; = \; A_y D B_y \; = \; \sum_{y'} A_{y'} D B_y \; = \; D B_y$. 
We can now apply Lemma~\ref{idemlemm}, taking $a = A_y$, $d = D$, and $b = B_y$, to conclude that 
$\EExiy$ and $\FFxy$ are projectors, and moreover that they commute with $D^2$. 
We can use the last part of
Lemma~\ref{idemlemm} to conclude that $\EExiy = \FFxyT$ whenever $x = \vx_i$.
\end{proof}

Finally, we show that the state can be chosen to be maximally entangled.
\begin{lemma}
\label{meslemm}
Let $(\psi', \{ \EE_{\vx} \}, \{ \FFx \})$ be a quantum perfect strategy where $\psi' = \sum_{i=1}^d \lambda_i e_i \otimes e_i$ with $\lambda_i > 0$ for all $i$. Then $(\psi, \{ \EE_{\vx} \}, \{ \FFx \})$ is a quantum perfect strategy, where $\psi = 1/\sqrt{d} \sum_{i=1}^d e_i \otimes e_i$ is the maximally entangled state.
\end{lemma}
\begin{proof}
Let $D'$ be the diagonal matrix associated with $\psi'$. Using Lemma~\ref{pvmtlemm}, 
\[ \psi'^*(\EExy \otimes \FFxy)\psi' = 0 \iff \EExy D' \FFxyT = \Ze \iff \EExy \FFxyT D' = \Ze \iff \EExy \FFxyT = \Ze , \]
since $D'$ commutes with $\FFxyT$ and is invertible. 
Similarly,
\[ \psi'^*(\EExy \otimes I)\psi' = 0 \iff \EExy D'  = \Ze  \iff \EExy = \Ze . \]
These conditions will be preserved by any state $\psi$ whose diagonal matrix is full rank and commutes with the matrices $\FFxy$. This holds in particular for the maximally entangled state, whose diagonal matrix has the form $\frac{1}{\sqrt{d}} I_d$.
\end{proof}

We can now combine Lemmas~\ref{fullranklemm}, \ref{pvmtlemm} and \ref{meslemm} to obtain the following result:
\begin{theorem}
\label{specformthm}
The  existence of a quantum perfect strategy implies the existence of a strategy $(\psi, \{ \EE_{\vx} \}, \{ \FFx \})$ with the following properties:
\begin{itemize}
\item The POVMs $\EEix$ and $\FFx$ are projective.
\item The state $\psi$ is a maximally entangled state $\psi = 1/\sqrt{d} \sum_{i=1}^d e_i \otimes e_i$.
\item If $x = \vx_i$ then $\EExiy = \FFxyT$.
%\item If $x = \vx_i$ and $y \neq y'$, then $\EExiy \FFxypT = \Ze$.
\item If $\vx \in \RA{}$ and $\vy \not\in \RB{}$, then $\EExy = \Ze$.
\end{itemize}
\end{theorem}

It is worth noting that the procedure for obtaining the strategy in this special form has three steps:
\begin{itemize}
\item In Step 1, the state and strategies are projected down to the subspace corresponding to the support of the Schmidt decomposition of the state. This step reduces the dimension of the Hilbert space, and preserves the probabilities for the strategy exactly.
\item Step 2 does not change the strategy at all, but shows that it must already have strong properties.
\item Step 3 changes the state but not the measurements. In general, the probabilities for the strategy will be changed, but the possibilities are preserved exactly.
\end{itemize}
Thus in passing to the special form, the dimension is reduced; the process by which we obtain projective measurements is not at all akin to dilation.

This theorem shows that all the information determining the strategy is in Alice's operators. Moreover, Alice's operators must be chosen non-contextually, so that $\EExiy$ is independent of the context $\vx$.
This means that we can define projectors $\Pxy := \EExiy$ whenever $x = \vx_i$. If $\vx_i = x = \vx'_j$, then we have $\EExiy = \FFxyT = \EExjy$, so $\Pxy$ is well-defined.

Now, recall the notion of joint measurability: 
a family of POVMs $\{ A^1_{y} \}_{y \in Y_1}$, \ldots, $\{A^k_{y} \}_{y \in Y_k}$
is said to be jointly measurable if there is a POVM $\{B_{\vy}\}_{\vy \in Y_1 \times \cdots \times Y_k}$
such that for all $i$, $A^{i}_{y} = \sum_{\vy_i = y} B_{\vy}$.
The following result is standard \cite{heinosaari2008notes}.
\begin{proposition}
\label{commprop}
A family of projective measurements is jointly measurable by a POVM if and only if they pairwise commute,
and in this case the POVM is the product of the family, and hence projective.
\end{proposition}

For each $\vx \in \RA{}$,
the projective measurements $\{ P_{\vx_i} \}$ defined above are jointly measured by $\EEx$.
Thus $\EEx$ is the projective measurement
$\Pvx = \{ \Pvxy \}_{\vy}$ given by $\Pvxy := P_{\vx_1, \vy_1} \cdots P_{\vx_k, \vy_k}$.

We shall now introduce the notion of \emph{quantum homomorphism} between relational structures $\As$ and $\Bs$. A quantum homomorphism is a family of projectors $\{ \Pxy \}_{x \in A, y \in B}$ in $\Projd$ for some $d$, satisfying the following conditions:
\begin{itemize}
\item (QH1) For all $x \in A$, $\sum_{y \in B} \Pxy \, = \, I$.
\item (QH2) For all $\vx \in \RA{}$,  $x = \vx_i$, $x' = \vx_j$, and $y, y' \in B$, $[\Pxy, \Pxyp] = \Ze$.
Thus we can define a projective measurement $\Pvx = \{ \Pvxy \}_{\vy}$, where $\Pvxy := P_{\vx_1, \vy_1} \cdots P_{\vx_k, \vy_k}$.
\item (QH3) If $\vx \in \RA{}$ and $\vy \not\in \RB{}$, then $\Pvxy = \Ze$.
\end{itemize}
Note that (QH1) implies that for any $x$, 
%the projectors $\Pxy$ and $\Pxnyp$ are orthogonal, \ie 
$\Pxy\Pxnyp = \Ze$ whenever $y \neq y'$.

We write $\As \arq \Bs$ for the existence of  a quantum homomorphism from $\As$ to $\Bs$.

\begin{theorem}
\label{qhthm}
For finite structures $\As$, $\Bs$, the following are equivalent:
\begin{enumerate}
\item There is a quantum perfect strategy for the homomorphism game from $\As$ to $\Bs$.
\item  $\As \arq \Bs$.
\end{enumerate}
\end{theorem}
\begin{proof}
The implication from (1) to (2) follows directly from Theorem~\ref{specformthm} and the subsequent discussion.
%Proposition~\ref{vecpvmprop}.
For the converse, given a quantum homomorphism $\{ \Pxy \}_{x \in A, y \in B}$, we can define $\EExy := \Pvxy$, $\FFxy := \Pxy^{T}$, and use the maximally entangled state to obtain a quantum perfect strategy. It is straightforward to verify that the homomorphism conditions (QH1)--(QH3) imply the strategy conditions (QS1) and (QS2).
\end{proof}

As a final remark, although we have focussed on a single relation to simplify the notation, our results go through for arbitrary relational signatures. Note that the general form of condition (QH2) is that $\Pxy$ and $\Pxyp$ must commute whenever $x$ and $x'$ are adjacent in the Gaifman graph of $\As$ --- that is, they both occur in some tuple of some relation.

\section{From quantum homomorphisms to the quantum monad}

We now show how to characterize quantum homomorphisms as the Kleisli morphisms of a monad $\Qd$ on the category of relational structures. This monad is \emph{graded} \cite{milius2015generic} by the dimension of the Hilbert space used as the quantum resource.

The monad will be defined on the category $\Rs$ of all $\sg$-structures, since $\Qd \As$ will always be infinite, even if $\As$ is finite. For the underlying universes of the structures, the construction  can be seen as a quantum variant of the discrete distribution monad \cite{jacobs2010convexity}, widely used in coalgebra and semantics. It is well known that the distribution monad can be defined over any commutative semiring, with the non-negative reals being used for the standard case of probabilities \cite{jacobs2010convexity,AbramskyBrandenburger}. Here we shall use the projectors $\Projd$ with $d$ ranging over the positive integers. For each $d$, $\Projd$ is a partial commutative semiring, since we can only add projectors if they are orthogonal, and only multiply them if they commute. We also have the graded multiplication given by the tensor product: if $P \in \Projd$ and $Q \in \Projdp$, then $P \otimes Q \in \Projddp$.
%\footnote{Note that, as we are working with matrices and using the Kronecker product, this operation is strictly associative, with $I_1$ as strict unit.}

We fix a relational signature $\sigma$.
For each positive integer $d$ and $\sg$-structure $\As$, we shall define a $\sg$-structure $\Qd \As$.
The universe of this structure $\Qd A$ is the set of all functions $p : A \rarr \Projd$ satisfying the 
normalization condition: $\sum_{x \in A} p(x) = I$.
%following two conditions:
%\begin{itemize}
%\item Finite support: $p(x) = \Ze$ for all but finitely many $x$.
%\item Normalization: $\sum_{x \in A} p(x) = I$.
%\end{itemize}
Note that the normalization condition implies that the projectors $\{ p(x) \}_{x \in A}$ are pairwise orthogonal.
Since we are in finite dimension $d$, this in turn implies that $p$ has finite support: $p(x) = \Ze$ for all but finitely many $x$.
We can think of $\Qd A$ as the projector-valued distributions on $A$ in dimension $d$.

For each relation $R$ of arity $k$ in $\sg$, we define $\RQd$ to be the set of all tuples $(p_1,\ldots ,p_k)$ such that:
\begin{itemize}
\item (QR1) For all $i, j \in [k]$, $x, x' \in A$: $[p_i(x), p_j(x')] = \Ze$.
\item (QR2) For all $\vx \in A^k$, if $\vx \not\in \RA{}$, then $\vp(\vx) = \Ze$, where $\vp(\vx) := p_1(x_1) \cdots p_k(x_k)$.
\end{itemize}
Note that the first condition implies that the product of projectors in the second is a well-defined projector.

\begin{proposition}
\label{qhKlprop}
Let $\As$ and $\Bs$ be finite $\sg$-structures.
There is a bijective correspondence between quantum homomorphisms $\{ \Pxy \}_{x \in A,y \in B}$ from $\As$ to $\Bs$ in dimension $d$ and standard homomorphisms $h : \As \rarr \Qd \Bs$.
\end{proposition}
\begin{proof}
Given a quantum homomorphism $\{ \Pxy \}$, define $h : \As \rarr \Qd \Bs$ by $h(x) = p$, where $p(y) := \Pxy$.
(QH1) implies normalization. Given $\vx \in \RA{}$, we have to show that $ \vp = h(\vx) \in \RQd$. (QH2) implies that (QR1) is satisfied for $p_i(y)$, $p_j(y')$, where $p_i = h(x_i)$, $p_j = h(x_j)$. Similarly, (QH3) implies (QR2).

For the converse, given $h : \As \rarr \Qd B$, define $\Pxy := h(x)(y)$. Again, normalization implies (QH1), (QR1) implies (QH2), and (QR2) implies (QH3).
\end{proof}
This correspondence is analogous to the familiar one between relations and set-valued functions, which shows that the category of relations is the Kleisli category of the powerset monad on $\Set$.

Now we show that $\Qd$ extends to a functor on $\CS$.
Given a homomorphism $h : \As \rarr \Bs$, we define $\Qd h: \Qd \As \rarr \Qd \Bs$ by $\Qd h(p)(y) := \sum_{h(x) = y} p(x)$. 
\begin{proposition}
$\Qd h$ is a well-defined homomorphism. Moreover, $\Qd$ is functorial: $\Qd g \circ \Qd h = \Qd (g \circ h)$ and $\Qd \id_A = \id_{\Qd A}$.
\end{proposition}
\begin{proof}
The finite support and normalization conditions ensure that $\Qd h$ is well-defined. 
Functoriality is proved exactly as for the distribution monad.
We verify that $\Qd h$ is a homomorphism. Suppose that $(p_1, \ldots , p_k) \in \RQd$.
By (QR1), this implies that $[p_i(x), p_j(x')] = \Ze$ for all $i, j \in [k]$, $x, x' \in A$. This implies that 
\[ [\Qd h(p_i)(y), \Qd h(p_j)(y')] \; = \; [\sum_{h(x) = y} p_i(x), \sum_{h(x') = y'} p_j(x')] \; = \; \Ze \]
so $h(\vp)$ satisfies (QR1). For (QR2), if $\vy \not\in \RB{}$, 
\[ h(\vp)(\vy) \; = \; \prod_{j=1}^k \sum_{h(\vx_j) \; = \; \vy_j} p_j(\vx_j) \; = \; \sum_{h(\vx) \; = \; \vy} \vp(\vx)  \; = \; \Ze,  \]
by (QR2) for $\vp \in \RQd$,
since $\vy \not\in \RB{}$ and $h(\vx) = \vy$ implies $\vx \not\in \RA{}$.
Thus $h(\vp) \in \RQdB$.
\end{proof}

The unit of the monad $\unA : \As \rarr \Qu \As$ sends $x \in A$ to the ``delta distribution'' $\delx \in \Qu A$, where $\delx(x) = I_1$, $\delx(x') = \Ze$ if $x \neq x'$. Verification that this is well-defined and yields a natural transformation is straightforward.

We also have the graded monad multiplication: $\muddA : \Qd \Qdp \As \rarr \Qddp \As$.
This is defined as follows: $\muddA(P)(x) := \sum_{p \in \Qdp \As} P(p) \otimes p(x)$.
We prove that this gives a well-defined natural transformation in the Appendix.
Thinking of $\otimes$ as the graded semiring multiplication on projectors, we can see the correspondence to the distribution monad. 

We recall that given a category $\CC$, the endofunctor category $[\CC, \CC]$ is monoidal; a monad on $\CC$ is a monoid in this category \cite{mac2013categories}.
Now let $(M, {\cdot}, 1)$ be a monoid, which we can view as a discrete category with a strict monoidal structure.
An \emph{$M$-graded monad} \cite{milius2015generic} on  $\CC$ is a lax monoidal functor from $M$ into $[\CC, \CC]$. 
Such a functor is given by the following data: an assignment $m \mapsto T_m$ of an endofunctor on $\CC$ to each element of $M$; 
%a natural transformation $\iota_{m, m'} : T_m \natarrow T_{m'}$ for each $m \leq m'$; 
a natural transformation $\eta : \Id \natarrow T_1$ (the graded unit); and a natural transformation $\mu^{m, m'} : T_m T_{m'} \natarrow T_{m \cdot m'}$ for all $m$, $m'$ (the graded multiplication). These are subject to coherence conditions, which generalize the usual monad equations. We refer to \cite{milius2015generic} for details.

In our case, we use the monoid $\Np$ of positive integers under multiplication.

\begin{theorem}
%The quadruple $(\{ \Qd \}_{d}, \{ \iddp \}_{d \leq d'}, \eta, \{ \mudd \}_{d,d'})$ 
The triple $(\{ \Qd \}_{d}, \eta, \{ \mudd \}_{d,d'})$ 
is a $\Np$-graded monad on $\CS$.
\end{theorem}

The proof of this result involves verifying a number of equations, and is fairly lengthy but straightforward.
We provide details in the Appendix.

We are particularly interested in the Kleisli category for this graded monad. The objects of this category are the same as those of $\CS$. A morphism from $\As$ to $\Bs$ is a homomorphism $h : \As \rarr \Qd \Bs$. 
By Proposition~\ref{qhKlprop}, we know that Kleisli morphisms correspond exactly to quantum homomorphisms.

%The graded composition of Kleisli arrows $h : \As \rarr \Qd \Bs$ and $k : \Bs \rarr \Qdp \Cs$ uses the graded Kleisli extension $\kdagd : \Qd \Bs \rarr \Qddp \Cs$ given by $\kdagd := \muddB \circ \Qd k$. We can then define $k \Kcomp h := \kdagd \circ h$.

The graded composition of Kleisli arrows $h : \As \rarr \Qd \Bs$ and $k : \Bs \rarr \Qdp \Cs$ 
is the arrow $k \Kcomp h : \As \rarr \Qddp \Cs$ given by $k \Kcomp h := \muddB \circ \Qd k \circ h$.
An explicit description can be calculated from the graded monad structure given above:
$(k \Kcomp h)(x)(z) \; = \; \sum_{y \in B} h(x)(y) \otimes k(y)(z)$.
If we write this in terms of the corresponding quantum homomorphisms $\{ \Pxy \}_{x \in A, y \in B}$, $\{ \Qyz \}_{y \in B, z \in C}$, we obtain  $\{ R_{x,z} \}_{x \in A, z \in C}$ given by the formula $R_{x,z} = \sum_{y \in B} \Pxy \otimes \Qyz$.
This generalizes the concrete definition for quantum graph homomorphisms from \cite{manvcinska2016quantum}.

%% We shall give an explicit description of this composition, which can be calculated from the graded monad structure given above.
%% \begin{proposition}
%% The graded Kleisli composition $k \Kcomp h : \As \rarr \Qddp \Cs$ is given by
%% $k \Kcomp h(x)(z) \; = \; \sum_{y \in B} h(x)(y) \otimes k(y)(z)$.
%% \end{proposition}
%% If we write this in terms of the corresponding quantum homomorphisms $\{ \Pxy \}_{x \in A, y \in B}$, $\{ \Qyz \}_{y \in B, z \in C}$, we obtain the formula $\sum_{y \in B} \Pxy \otimes \Qyz$.
%% This generalizes the concrete definition for quantum graph homomorphisms given in \cite{manvcinska2016quantum}.

\section{Quantum advantage via the quantum monad}

We shall now show how the quantum monad provides a unified framework for expressing quantum advantage in a wide range of information processing tasks. We shall show equivalences between:
\begin{itemize}
\item state-independent strong contextuality arguments
\item quantum advantage in constraint satisfaction
\item existence of quantum (but not classical) homomorphisms between relational structures.
\end{itemize}

\subsection{Classical correspondences}

We begin with the standard classical correspondence between constraint satisfaction problems and the existence of homomorphisms.
A CSP instance has the form $\CSP = (V, D, C)$ where $V$ is a set of variables, $D$ is a domain of 
values\footnote{One could have different domains associated with different variables. However, this is an inessential generalization, which we omit to keep notation simple.},  and $C$ is a set of constraints of the form $(\vx, r)$, where for some $k$, $\vx \in V^k$, and $r \subseteq D^k$.
We say that $c = (\vx, r)$ is a $k$-ary constraint. A \emph{solution} of the CSP is a  function $s : V \rarr D$ such that, for all $(\vx, r) \in C$, $s(\vx) \in r$, where $s(\vx) := (s(x_1), \ldots , s(x_k))$.

Given $\CSP = (V, D, C)$ we define two structures over the signature with a $k$-ary relation symbol $R_c$ for each $k$-ary constraint $c$. First, $\BK$ has as universe $D$, and for each $c = (\vx, r) \in C$, $\RcB = r$. Secondly, $\AK$ has universe $V$, and for each $c = (\vx, r) \in C$, $\RcA = \{ \vx \}$.
The following is immediate:
\begin{proposition}
There is a one-to-one correspondence between solutions for $\CSP$ and homomorphisms $\AK \rarr \BK$.
\end{proposition}

There is also a converse to this result. Given $\sg$-structures $\As$ and $\Bs$, we can define the CSP $\KAB = (V, D, C)$, where $V = A$, $D = B$, and $C = \{ (\va, \RB{}) \mid  R \in \sg, \; \va \in \RA{} \}$.
\begin{proposition}
There is a one-to-one correspondence between homomorphisms $\As \rarr \Bs$ and solutions for $\KAB$.
\end{proposition}

We will now look at \emph{empirical models over measurement scenarios}, introduced in \cite{AbramskyBrandenburger} as a general setting for studying contextuality, in quantum mechanics and beyond, with non-locality as a special case.
%This is also known as the sheaf-theoretic approach to contextuality, but we will not emphasize the sheaf aspects here. 

A measurement scenario is a triple $(X, \MM, O)$, where $X$ is a set of measurement labels; $\MM$ is a family of subsets of $X$, where we think of $C \in \MM$ as a set of compatible measurements, or a context; and $O$ is a set of measurement outcomes. An empirical model $e : (X, \MM, O)$ for a scenario is a family $e  = \{ \ec \}_{C \in \MM}$ of probability distributions $\ec \in \Prob(O^C)$ on the joint outcomes of measuring all the variables in a context $C$. Such empirical models can arise from observational data (hence the name), or be generated by measuring a quantum state in contexts comprising jointly measurable observables. A hierarchy of notions of contextuality can be defined in this general setting \cite{AbramskyBrandenburger,AbramskyEtAl:ContextualityCohomologyAndParadox}. We shall be concerned with \emph{strong contextuality}. We say that $e : (X, \MM, O)$ is strongly contextual if there is no global assignment $g : X \rarr O$ such that, for all $C \in \MM$, $\ec(g |_{C}) > 0$. That is, there is no global assignment which is consistent with the model in the sense of yielding \emph{possible} outcomes (non-zero probability)  in all contexts. This form of contextuality is witnessed by the GHZ construction \cite{GHZ,Mermin90:QuantumMysteriesRevisited-SimplifiedGHZ1}, as well as by Kochen--Specker paradoxes \cite{kochen1975problem}, and post-quantum devices such as the PR box \cite{popescu1994quantum}.

Given $e : (X, \MM, O)$, we fix an ordering on $X$, and define a CSP $\Ke = (X, O, C)$, where $C$ is the set of constraints $((x_1, \ldots , x_k), r)$ such that  $x_1 < \cdots < x_k$, $\{ x_1, \ldots , x_k \} \in C$ and $r = \{ s(\vx) \mid \ec(s) > 0 \}$.
\begin{proposition}
There is a one-to-one correspondence between consistent global assignments for $e$ and solutions of $\Ke$. Thus $e$ is strongly contextual iff $\Ke$ has no (classical) solution.
\end{proposition}

We thus have a three-way correspondence between CSPs, empirical models, and homomorphisms between relational structures.

\subsection{Quantum solutions}

We now consider how quantum resources enter the picture. Since we already have a notion of quantum homomorphism for general relational structures, the correspondences established in the previous subsection give us ready-made notions of quantum solution for CSPs and empirical models.
We define a \emph{quantum solution} for a CSP $\CSP = (V, D, C)$ to be a quantum homomorphism $\AK \arq \BK$, \ie a Kleisli morphism $\AK \rarr \Qd \BK$ for some $d$.
Similarly, we define a quantum solution for an empirical model $e : (X, \MM, O)$ to be a quantum homomorphism $\AKe \arq \BKe$.

We shall now compare these notions to existing ones for empirical models and constraints. These will turn out to be special cases.

Note first that given an empirical model $e$, the corresponding CSP $\Ke$ is determined purely by the \emph{supports} of the probability distributions $\ec$, \ie the possibilistic content of the model.
It is only this information which is relevant to strong contextuality.
We can consider a fine-grained notion of realization of a probabilistic empirical model, as in \cite{AbramskyBrandenburger}. However, if our concern is to give a quantum witness for strong contextuality of 
 an empirical model $e : (X, \MM, O)$, then this is given by a state $\psi$, and a PVM $P_x = \{ \Pxo \}_{o \in O}$ for each $x \in X$, such that $[\Pxo, \Pxpop] = \Ze$ whenever $x$ and $x'$ both occur in some $C \in \MM$. These must then satisfy, for all $C \in \MM$ and $s \in O^C$, $\ec(s) = 0 \, \Rightarrow \, \psi^*P_{\vx,s(\vx)}\psi = 0$, where $P_{\vx,s(\vx)} = P_{x_{1},s(x_{1})} \cdots P_{x_{k},s(x_{k})}$.

This provides a \emph{quantum witness} for strong contextuality if $e$ is a strongly contextual empirical model.
An example is provided by the GHZ state, using $X$ and $Y$ measurements for each party \cite{AbramskyBrandenburger}. 
An infinite family of such examples using three-qubit states is described in \cite{tripartite}.
    
We can also consider a stronger notion. A \emph{state-independent} quantum witness for $e : (X, \MM, O)$ is given by a family of PVMs $\{ P_x \}_{x \in X}$ which, for \emph{any} state $\psi$, yield a quantum witness for $e$.
The Mermin magic square and pentagram \cite{Mermin90:SimpleUnifiedForm}, and Kochen--Specker constructions \cite{kochen1975problem,cabello1996bell}, provide examples of state-independent quantum realizations of strong contextuality.
Note that in the state-independent case, we have the condition: $\ec(s) = 0 \Rightarrow \Pvxvo = \Ze$.
Comparison with the definition of quantum homomorphism $\AKe \arq \BKe$ immediately yields the following result:
\begin{proposition}
For an empirical model $e : (X, \MM, O)$ there is a one-to-one correspondence between state-independent quantum witnesses for $e$ and quantum solutions for $\Ke$.
\end{proposition}

An interesting point arising from this result, taken together with the results from Section~\ref{qpstoqhsec}, is that state-independent strong contextuality proofs can always be underwritten by non-locality arguments. This can be seen as a general form of constructions for turning Kochen--Specker contextuality proofs into Bell non-locality arguments \cite{HeywoodRedhead83:Nonlocality}. Indeed, the r\^ole of the entangled state and of Bob in the non-local game is to provide an operational or physical underpinning for the compatibility or generalized no-signalling assumption which is made for empirical models \cite{AbramskyBrandenburger}.
Can we find a similar underpinning in the state-dependent case? We shall return to this point in the final section.

%\subsubsection*{Binary constraint systems}

We now consider \emph{binary constraint systems} (BCS), which have been extensively studied \cite{cleve2014characterization,cleve2017perfect,slofstra2016tsirelson,ji2013binary}.
We shall follow the account in \cite{cleve2014characterization}.
A  BCS $(V, C)$ is simply a boolean CSP $(V, \{ 0, 1 \}, C)$. In this case, constraints can be  written in the form $c = (\vx, \bc)$, where $\bc : \{ 0, 1 \}^k \rarr \{ 0, 1 \}$ is a boolean function.
Our general notion of quantum solution yields in this case a family of PVMs  $\Px = \{ \Pxo \}_{o \in \{ 0,1\}}$ for $x \in V$, such that $[\Pxo, \Pxpop] = \Ze$ whenever $x$ and $x'$ both occur in some constraint, and $\Pvxvo = \Ze$ for $c = (\vx, \bc)$ with $\bc(\vo) = 0$.

In \cite{cleve2014characterization}, a notion of \emph{operator solution} for a BCS is defined.
This is an assignment of a self-adjoint operator (aka observable) $\Ax$ to each variable $x$, such that: (1) each $\Ax$ is binary, \ie $\Ax^2 = I$; (2) $[\Ax, \Axp] = \Ze$ when $x$ and $x'$ both occur in the same constraint.
To express constraints, the representation of boolean values $b \mapsto (-1)^b$, $b \in \{0,1\}$,  is used. It is standard that each boolean function $\{ -1, +1 \}^k \rarr \{ -1, +1 \}$ can be uniquely represented by a real multilinear polynomial $p(X_1,\ldots , X_k)$ \cite{o2014analysis}. Moreover, if the corresponding boolean function in the $\{ 0,1\}$-representation is $b : \{ 0,1 \}^k \rarr \{ 0, 1\}$, then for $\vo \in \{ 0,1\}^k$: 
$(-1)^{b(\vo)} = p((-1)^{\vo_1},\ldots,(-1)^{\vo_k})$.
It is also standard that if we substitute pairwise commuting self-adjoint operators $A_1$, \ldots , $A_k$ for the variables $X_1,\ldots,X_k$, we obtain a self-adjoint operator $p(A_1,\ldots,A_k)$. The condition for an operator solution to satisfy the constraints is then expressed as follows: for each constraint $c = (\vx, \bc)$, where $\bc$ is represented by the polynomial $\pc(X_1, \ldots , X_k)$, we must have $\pc(A_1,\ldots,A_k) = -I$.

To relate this notion of operator solution to our quantum solution, note that a binary observable $A_x$ with $A_{x}^2 = I$ has a spectral decomposition $A_x = P_{x,0} - P_{x,1}$, where $P_{x,0}$, $P_{x,1}$ are projectors. Commutation of the observables for variables occurring in the same context is equivalent to commutation of the corresponding projectors.
Given a constraint $c = ((x_1,\ldots,x_k), \pc)$, since the observables $A_{x_1}$, \ldots , $A_{x_k}$ pairwise commute, we obtain a resolution of the identity $\sum_{\vo \in \{ 0,1\}^k} \Pvxvo =I$. Multiplying 
$\pc(A_1,\ldots,A_k)$ by this expression yields
\[ \pc(A_1,\ldots,A_k)  \;\; = \;\;  \sum_{\vo \in \{ 0,1\}^k} \pc((-1)^{\vo_1},\ldots,(-1)^{\vo_k}) \Pvxvo . \]
It follows that $\pc(A_1,\ldots,A_k) = -I$ iff for all $\vo$ with $\bc(\vo) = 0$, $\Pvxvo = \Ze$.
As an immediate consequence, we have:
\begin{proposition}
Given a BCS $(V, C)$, there is a one-to-one correspondence between operator solutions of the BCS and quantum solutions of the corresponding CSP.
\end{proposition}

\subsection{Graphs}

The results we have seen thus far show that our notion of quantum homomorphism subsumes a number of existing notions in contextuality and non-local games. However, as we shall now see, the situation in the setting which provided the original motivation for our approach, namely graph homomorphisms, is somewhat more subtle.

Graphs arise as structures for the signature with a single binary relation. Simple graphs are those where the relation is symmetric and irreflexive. If we specialize our definition of quantum homomorphism to the case $G \arq H$ between graphs $G$ and $H$, this gives a family $\{P_{x,y}\}_{x \in V(G), y \in V(H)}$
of projectors satisfying the following conditions:
\begin{itemize}
\item for all $x, x' \in V(G)$ and $y, y' \in V(H)$ with $x \sim x'$, $[P_{x,y},P_{x',y'}] = 0$;
\item for all $x, x' \in V(G)$ and $y, y' \in V(H)$ with $x \sim x'$ and $y \not\sim y'$, $P_{x,y}P_{x',y'} = 0$.
\end{itemize}

This definition differs from that introduced by Man\v{c}inska and Roberson \cite{manvcinska2016quantum}
as a generalization of quantum graph colouring \cite{cameron2007quantum},
in that the latter does \emph{not} impose the first condition forcing commutativity between the PVMs corresponding to adjacent vertices of $G$.
We refer to this as an \emph{MR quantum graph homomorphism} in order to distinguish it from our notion.

This reflects a difference in the Alice--Bob game used to motivate each definition.
In the game we consider,
Alice receives an ordered edge of $G$ as an input and Bob a vertex of $G$,
and we require that Alice answers with an edge of $H$ and Bob with a vertex of $H$
in a fashion consistent with Alice's choice.
By contrast, Man\v{c}inska and Roberson consider a symmetric Alice--Bob game,
where each player receives
a vertex of $G$ as input and outputs a vertex of $H$,
and their outputs are required to be the same when they receive the same vertex,
and required to form an edge of $H$ whenever their inputs form an edge of $G$.

It is clear that a quantum homomorphism between graphs in the sense of this paper
is also an MR quantum graph homomorphism.
But the precise relationship between the two notions is yet to be understood in general:
in particular, whether the existence of an MR quantum graph homomorphism 
implies the existence of a quantum homomorphism in our sense.
%\footnote{A result by Ji \cite{ji2013binary} indicates that the two notions coincide in the case of $3$-colourability.}
However, by adapting a construction due to Ji \cite{ji2013binary},
we can capture the existence of MR quantum graph homomorphisms
in terms of quantum homomorphisms of relational structures,
via a BCS.

Given graphs $G$ and $H$, we define $V = \{ \rxy \mid x \in V(G), y \in V(H) \}$. For each $x \in V(G)$, we have a boolean constraint $\bigvee_y \rxy$; a constraint $\neg (\rxy \, \wedge\,  \rxyp)$ when $y \neq y'$; and a constraint $\neg (\rxy \wedge \rxpyp)$ whenever $x \sim x'$ and 
$y \not\sim y'$. This defines a BCS $(V, C)$.

\begin{theorem}
Given graphs $G$ and $H$, there is a one-to-one correspondence between MR quantum graph homomorphisms from $G$ to $H$ and quantum homomorphisms $\As_{\CSP} \arq\Bs_{\CSP}$ for the associated CSP $\CSP = (V, \{ 0, 1\}, C)$.
\end{theorem}
\begin{proof}
We recall that an MR quantum graph homomorphism is given by a family of projectors $\{ \Pxy \}_{x \in V(G), y \in V(H)}$ such that (MR1) $\sum_y \Pxy = I$, and (MR2) $\Pxy\Pxpyp = \Ze$ whenever $x \sim x'$ and $y \not\sim y'$.
A quantum homomorphism $\As_{\CSP} \arq \Bs_{\CSP}$ is given by a family of projectors $\{ \Qxyo \}$, $x \in V(G)$, $y \in V(H)$, $o \in \{0,1\}$, such that the following conditions hold: (QH1) $\Qxyf + \Qxyt = I$; (QH2) $\Qxyt \Qxypt = \Ze$, $(y \neq y')$; (QH3) $\Qxyt \Qxpypt = \Ze$, $(x \sim x' \wedge y \not\sim y')$; (QH4) $Q_{xy_{1},0} \cdots Q_{xy_{k},0} = \Ze$, $V(H) = \{ y_1, \ldots , y_k \}$.
The commutativity conditions which are additionally required are implied by the orthogonality conditions (QH2) and (QH3).

Given an MR homomorphism $\{ \Pxy \}$, we define $\Qxyt := \Pxy$, $\Qxyf := I - \Pxy$. Clearly (QH1) is satisfied. (MR1) implies (QH2), while (MR2) implies (QH3).
Finally, using (QH2), $(I - P_{x,y_{1}}) \cdots (I - P_{x,y_{k}}) = I - \sum_y \Pxy$, and by (MR1), (QH4) holds. 

Conversely, given a quantum homomorphism $\{ \Qxyo \}$, we define $\Pxy := \Qxyt$. (MR2) follows from (QH3),
while using (QH1), we can reverse the reasoning in the previous paragraph to show that (QH4) implies (MR1).
These passages are clearly mutually inverse, so the result follows.
\end{proof}

It is noteworthy that our approach allows us to avoid \textit{ad hoc} coding of constraints by polynomials, as in \cite{cleve2014characterization,ji2013binary}. Instead, we quantize the standard classical notions in a uniform way, using the quantum monad.

\section{Outlook}
This work suggests a number of directions for further study. We list a few:
\begin{itemize}
\item A notion of quantum graph isomorphism, with an equivalent characterization via an Alice--Bob game, has been studied in \cite{atserias2016quantum}. This can be generalized to relational structures.
How does this fit in our quantum monad framework? 
\item Our approach captures quantum advantage provided by state-independent strong contextuality. Does state-dependent contextuality admit a similar treatment?
\item Any strategy for an Alice--Bob game has a winning probability, which is related to the contextual fraction \cite{AbramskyBarbosaMansfield17:CF}. Can our approach be adapted to deal with quantitative aspects?
\item Homomorphisms are intimately related with the existential positive fragment.
Can this be extended to provide a notion of quantum validity for first-order formulae?
\item Can other concepts from finite model theory,
such as pebble games, which admit a comonadic formulation \cite{abramsky2017pebbling},  be similarly quantized?
\item The  algebras of the quantum monad can be described as convex structures with mixing weighted by projectors rather than just numbers in $[0,1]$. Is this viewpoint useful?
\end{itemize}

\bibliographystyle{plain}
\bibliography{qm} 

\section*{Appendix}

\subsection*{Review of linear algebra and quantum mechanics background}

Since we are working in finite dimensions, we will use standard matrix notation. Thus we are dealing with $d \times d'$ complex matrices.  We write matrix transpose as $A^T$. Apart from the usual operations of matrix addition and multiplication, there is the adjoint $A^*$, which is the conjugate transpose of $A$. Thus $[a_{i,j}]^{*} = [ \overline{a_{j,i}}]$. The zero matrix is $\Ze$, the identity matrix in dimension $d$ is $I = I_d$.\footnote{We will omit dimensional subscripts whenever we can get away with it.}
We view $d' \times d$ complex matrices interchangably as linear maps $\Complex^{d} \rarr \Complex^{d'}$, acting on ``column vectors'', \ie $d \times 1$ matrices. We identify $1 \times 1$ matrices with scalars, \ie complex numbers. The inner product of vectors $\vx, \vy \in \Complex^d$ is given  by the matrix product $\vx^* \vy$. The norm of a vector $\vx$ is $\| \vx \| := \sqrt{\vx^*\vx}$. The standard basis vectors in dimension $d$ are $e_1, \ldots , e_d$, where $e_i$ has $i$'th component $1$, and all other components $0$.

A square matrix $A$ is self-adjoint (aka Hermitian) if $A^* = A$. It is positive semidefinite if it self-adjoint, and $\vx^*A\vx \geq 0$ for all vectors $\vx$. If $A$ is positive semidefinite, so is $C^*AC$ for any $C$. We have the order $A \leq B$ if $B - A$ is positive semidefinite. The condition $A \geq \Ze$ says exactly that $A$ is positive semidefinite. The matrices $A \geq \Ze$ form a convex cone.
$A$ is a projector if $A^* = A = A^2$.  We write $\Projd$ for the set of $d \times d$ projectors.
A fact we shall use frequently is that for any family of projectors $\{ P_i \}$ in $\Projd$, $\sum_i P_i \leq I$ iff $P_iP_j = \Ze$ whenever $i \neq j$.

If $A = [a_{i,j}]$ is a $m \times n$ matrix and $B$ a $p\times q$ matrix, then the Kronecker product
$A \otimes B := [a_{i,j} B]$ is an $mp \times nq$ matrix, which represents the tensor product  of the corresponding linear maps. This operation is strictly associative, with unit $\mathbf{1} := [1]$. The key equation is the interchange law with matrix multiplication: $(A \otimes B)(C \otimes D) = AC \otimes BD$. This is functoriality.
The category of complex matrices is a strict monoidal category with respect to this operation.
Indeed, the category of complex matrices is equivalent to the category of finite-dimensional Hilbert spaces at the level of dagger compact closed categories, the basic setting for categorical quantum mechanics \cite{abramsky2004categorical,abramsky2008categorical}.
The final operation we consider is vectorization of a matrix: $\vct{A}$ turns a $d \times d'$ matrix into a $dd'$-vector by stacking the columns of $A$ on top of each other. In terms of the closed structure on the category of matrices, it is the name of the morphism $A$. The ``cup'' or unit of the compact closed structure is $\vct{I}$. We have the equation $\vct{A} = (I \otimes A)\vct{I}$, and the ``sliding rule'': $(A \otimes I)\vct{I} = (I \otimes A^{T})\vct{I}$, from which we can derive the key equation for vectorization: $(A \otimes B)\vct{C} = \vct{BCA^{T}}$.
Diagrammatically, this is
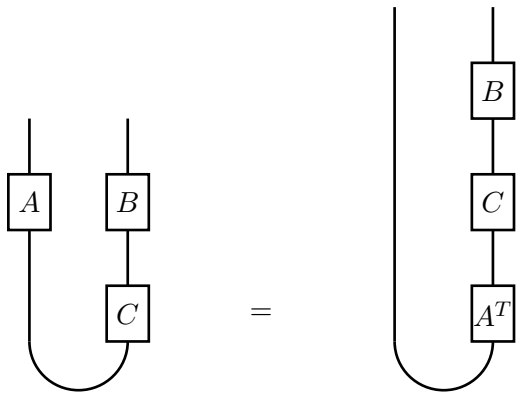
\begin{center}
\input{diag}

\end{center}
A vector $\psi \in \Complex^{d_A} \otimes \Complex^{d_B}$ has a \emph{Schmidt decomposition} $\psi = \sum_{i=1}^d \lambda_i \alpha_i \otimes \beta_i$, where $d \leq \min(d_a,d_b)$, $\{ \alpha_i \}$ and $\{ \beta_i \}$ are orthonormal sets of vectors in $\Complex^{d_A}$ and $\Complex^{d_B}$ respectively, and $\lambda_i > 0$ for all $i$. This follows directly from Singular Value Decomposition. We refer to $d$ as the \emph{Schmidt rank} of $\psi$.

The following result is standard, but we did not find an explicit reference so we include a proof.
\begin{proposition}
\label{trprop}
Let $A$ and $B$ be positive semidefinite matrices. Then $\Tr(AB) = 0 \iff AB = \Ze$.
\end{proposition}
\begin{proof}
Since $A$ and $B$ are positive semidefinite, they have positive semidefinite square roots $\Asq$, $\Bsq$, and
$\Tr(AB) = \Tr(\Asq\Asq\Bsq\Bsq) = \Tr(\Bsq\Asq\Asq\Bsq) = \Tr(C^*C)$, where $C := \Asq\Bsq$.
Since $C^*C$ is positive semidefinite, $\Tr(C^*C) = 0 \iff C^*C =\Ze \iff C = \Ze$. But $C = \Ze$ implies $AB = \Asq C\Bsq = \Ze$.
\end{proof}

Now we briefly review the needed notions from quantum mechanics. A (pure) state in dimension $d$ is a vector of unit norm in $\Complex^d$. A POVM (positive operator-valued measure) is a family  $\{ A_i \}_{i}$ with $A_i \geq \Ze$ for all $i$, and $\sum_i A_i = I$. The indices $i$ label the measurement outcomes. Measuring a POVM  $\{ A_i \}_{i}$ on a state $\psi$ yields outcome $i$ with probability $\psi^*A_i\psi$.
A POVM is projective (or a PVM) if $A_i$ is a projector for all $i$. This implies that $A_iA_j = \Ze$ for all $i \neq j$, \ie the projectors are mutually orthogonal. The product of projectors is a projector if and only if they commute, which is usually written as $[P, Q] = \Ze$, where $[P, Q] := PQ - QP$.

\subsection*{Quantum monads}

\begin{proposition}
For each $\As$, $\muddA$ is a well-defined homomorphism, and yields a natural transformation.
\end{proposition}
\begin{proof}
Firstly, $\muddA$ is a well-defined function, by finiteness of support. To show that it is a homomorphism,
consider  $(P_1,...,P_k) \in \RQdQdp$. We must show firstly that for all  $z,z' \in C$, and $i,j$, $[p_i(z), p_j(z')] = \Ze$, where $p_i := \muddA(P_i)$, $p_j := \muddA(P_j)$. Using linearity, this reduces to showing that
$P_i(p) \otimes p(z)$ commutes with $P_j(p') \otimes p'(z')$ for all $p, p' \in \Qdp$. Applying (QR1) to $P_i$ and $P_j$, we have that $P_i(p)$ commutes with $P_j(p')$. If $p(z)$ commutes with $p'(z')$, we are done.
If not, then we know that $p$ cannot be adjacent to $p'$ in the Gaifman graph of $\Qdp$. Hence for any expansion $\vp = (p_1, \ldots , p_k)$ with $p_i = p$, $p_j = p'$, $\vp \not\in \RQdp$, so by (QR2) we must have $P_1(p_1) \cdots P_k(p_k) = \Ze$. Using normalization, we have
$P_i(p)P_j(p') = \sum_{\vp_i =p, \vp_j = p'} P_1(p_1) \cdots P_k(p_k) = \Ze$, and so 
\[ (P_i(p) \otimes p(z))(P_j(p') \otimes p'(z')) = \Ze = (P_j(p') \otimes p(z'))(P_i(p) \otimes p(z)) . \]
Now let $q_i = \muddA(P_i)$, $i = 1,\ldots, k$. 
To show that (QR2) holds for $(q_1,\ldots,q_k)$ reduces similarly to showing that, if $\vx \not\in \RA{}$, then 
$P_1(p_1)\cdots P_k(p_k) \otimes p_1(x_1)\cdots p_k(x_k) = \Ze$ for all $p_1,\ldots , p_k$. If 
$p_1(x_1)\cdots p_k(x_k) = \Ze$ we are done; otherwise, since $\vx \not\in \RA{}$, we must have 
$\vp \not\in \RQdp$, and applying (QR2) to $(P_1,...,P_k)$, we must have $P_1(p_1)\cdots P_k(p_k) = \Ze$.
 
Naturality is commutativity of the following square. 
\begin{diagram}[4em]
\Qd \Qdp \As & \rTo^{\muddA} & \Qddp \As \\
\dTo^{\Qd \Qdp f} & & \dTo_{\Qddp f} \\
\Qd \Qdp \Bs & \rTo_{\muddB} & \Qddp \Bs
\end{diagram}
This is the following calculation:
\begin{align*}
\Qddp f \circ \muddA(P)(y) 
= \;\; & \sum_{f(x)=y} \sum_{p} P(p) \otimes p(x) \\
= \;\; & \sum_{q} \sum_{\Qdp f(p) = q} P(p) \otimes q(y) \\
= \;\; & \sum_{q} \Qd \Qdp(f)(P)(q) \otimes q(y) \\
= \;\; & \muddB \circ \Qd \Qdp f(P)(y) .
\end{align*}
The second step uses the fact that $\Qdp f(p) = q \iff q(y) = \sum_{f(x)=y} p(x)$.
\end{proof}

The unit $\eta:\Id \natarrow T_1$ and graded multiplication $\mu^{m,m'}:T_m T_{m'}\natarrow T_{m\cdot m'}$ of a graded $M$-monad are required to satisfy the following coherence conditions:
\begin{diagram}[width=5em]
T_mX &\rTo^{\eta_{T_mX}}&T_1T_mX&&T_mT_{m'}T_{m^{\prime\prime}}X&\rTo^{T_m\mu_X^{m',m^{\prime\prime}}}&T_mT_{m'\cdot m^{\prime\prime}}X\\
\dTo^{T_m\eta_X} & \rdEq &\dTo_{\mu_X^{1,m}}&&\dTo^{\mu_{Tm^{\prime\prime}X}^{m,m'}} &       &\dTo_{\mu_X^{m,m'\cdot m^{\prime\prime}}}\\
T_mT_1X& \rTo_{\mu_X^{m,1}}   &  T_mX&&T_{m\cdot m'}T_{m^{\prime\prime}}X& \rTo_{\mu_X^{m\cdot m',m^{\prime\prime}}}   &  T_{m\cdot m'\cdot m^{\prime\prime}}X.
\end{diagram} 

We verify these for the quantum monad.

\begin{lemma}\label{lem:1}
Let $\As$ be a structure and $d\in\Np$. Then, the following diagram commutes:
\begin{diagram}[width=6em]
\Qd\As&\relax\rTo^{\eta_{\Qd\As}}&\Qone\Qd\As\\
\dTo^{\Qd\eta_\As} & \rdEq  &\relax\dTo_{\mu_\As^{1,d}}\\
\Qd\Qone\As& \relax\rTo_{\mu_\As^{d,1}}   &  \Qd\As
\end{diagram}
\end{lemma}
\begin{proof}
Let $p\in \Qd\As$ and $x\in A$. The claim is that $(\mu_\As^{1,d}\circ\eta_{\Qd\As})(p)(x)=p(x) = (\mu_\As^{d,1}\circ \Qd\eta_\As)(p)(x)$.  The left-hand side of this equation expands to 
\begin{align*}
\mu_\As^{1,d}(\eta_{\Qd\As}(p))(x)
= \sum_{p'\in \Qd\As}\eta_{\Qd\As}(p)(p')\otimes p'(x)
= p(x),
\end{align*}
and the right-hand side to
\begin{align*}
\mu_\As^{d,1}( \Qd\eta_\As(p))(x)
  = \; & \sum_{p'\in \Qone\As} \Qd\eta_\As(p)(p')\otimes p'(x)
\\= \; &~ \Qd\eta_\As(p)(\delta_x)
\\= \; & \sum_{\eta_\As(x')=\delta_x}\eta_\As(p)(x')
\\= \; &~ p(x),
\end{align*}
where $\delta_{x}:A\to \{0,1\}$ is defined by $\delta_x(x'):=\delta_{x,x'}$ for all $x'\in A$. 
\end{proof}

\begin{lemma}\label{lem:2}
Let $\As$ be a structure and $a,b,c\in\Np$. Then, the following diagram commutes:
\begin{diagram}[width=6em]
\Qa\Qb\Qc\As&\relax\rTo^{\Qa\mu_\As^{b,c}}&\Qa\Qbc\As\\
\dTo^{\mu_{\Qc\As}^{a,b}} &       &\relax\dTo_{\mu_\As^{a,bc}}\\
\Qab\Qc\As& \relax\rTo_{\mu_\As^{ab,c}}   &  \Qabc\As
\end{diagram}
\end{lemma}
\begin{proof}
Let $P\in \Qa\Qb\Qc\As$ and $x\in A$. The claim is that $(\mu_\As^{a,bc}\circ \Qa\mu_\As^{b,c})(P)(x)=(\mu_\As^{ab,c}\circ \mu_{\Qc\As}^{a,b})(P)(x)$.
We have:
\begin{align*}
       &~ \mu_\As^{a,bc}(\Qa\mu_\As^{b,c}(P))(x)
\\= \; &\sum_{q\in \Qbc\As} \Qa\mu_\As^{b,c}(P)(q) \otimes q(x)
\\= \; &\sum_{q\in \Qbc\As}~\left(\sum_{\mu_\As^{b,c}(p')=q}P(p') \right) \otimes q(x)
\\= \; &\sum_{q\in \Qbc\As}~\sum_{\mu_\As^{b,c}(p')=q}P(p')\otimes \mu_\As^{b,c}(p')(x)
\\= \; &\sum_{q\in \Qbc\As}~\sum_{\mu_\As^{b,c}(p')=q}P(p')\otimes \left(\sum_{p\in \Qc\As}p'(p)\otimes p(x)\right)
\\= \; &\sum_{q\in \Qbc\As}~\sum_{\mu_\As^{b,c}(p')=q}~\sum_{p\in \Qc\As} P(p')\otimes p'(p)\otimes p(x)
\\= \; &\sum_{p'\in \Qb(\Qc\As)}~\sum_{p\in \Qc\As} P(p')\otimes p'(p)\otimes p(x)
\\= \; &\sum_{p\in \Qc\As}~\left(\sum_{p'\in \Qb(\Qc\As)} P(p')\otimes p'(p) \right)\otimes p(x)
\\= \; &\sum_{p\in \Qc\As}\mu_{\Qc\As}^{a,b}(P)(p)\otimes p(x)
\\= \; &~ \mu_\As^{ab,c}(\mu_{\Qc\As}^{a,b}(P))(x)
\end{align*}
\end{proof}

\subsection*{Additional material on the quantum monad}

We shall now show that the quantum monad is monoidal (or commutative) and affine \cite{kock1970monads}.
This continues the analogy with the distribution monad, which is well known to have these properties \cite{jacobs2010convexity}.

The category $\CS$ has finite products, given by the usual cartesian product of structures. The terminal object $\top$ is the one-element structure, with each relation interpreted as the universal relation.
Because of normalization, the following is immediate:
\begin{proposition}
For all $d$, $\Qd \top \cong \top$. Thus the quantum monad is affine.
\end{proposition}

Now given structures $\As$ and $\Bs$, we define a map $\tddpAB : \Qd \As \times \Qdp \Bs \rarr \Qddp (A \times B)$ by $\tddpAB(p,q)(x,y) := p(x) \otimes q(y)$.

\begin{proposition}
This is a well-defined homomorphism, and the family $\{ \tddpAB \}$ defines a graded  natural transformation satisfying the monoidal coherence conditions, thus witnessing a commutative strength.
\end{proposition}

\end{document}

%% file: diag.tex
\begin{tikzpicture}[scale=2.54]
% dpic version 2011.03.17 option -g for TikZ and PGF 1.01
\ifx\dpiclw\undefined\newdimen\dpiclw\fi
\global\def\dpicdraw{\draw[line width=\dpiclw]}
\global\def\dpicstop{;}
\dpiclw=0.8bp
\dpiclw=1bp
\dpicdraw (0,-0.145455) rectangle (0.218182,0.145455)\dpicstop
\draw (0.109091,0) node{$A$};
\dpicdraw (0.509091,-0.145455) rectangle (0.727273,0.145455)\dpicstop
\draw (0.618182,0) node{$B$};
\dpicdraw (0.509091,-0.727273) rectangle (0.727273,-0.436364)\dpicstop
\draw (0.618182,-0.581818) node{$C$};
\dpicdraw (0.109091,-0.727273)
 ..controls (0.109091,-0.867854) and (0.223055,-0.981818)
 ..(0.363636,-0.981818)
 ..controls (0.504218,-0.981818) and (0.618182,-0.867854)
 ..(0.618182,-0.727273)\dpicstop
\dpicdraw (0.618182,-0.436364)
 --(0.618182,-0.145455)\dpicstop
\dpicdraw (0.109091,0.145455)
 --(0.109091,0.436364)\dpicstop
\dpicdraw (0.618182,0.145455)
 --(0.618182,0.436364)\dpicstop
\dpicdraw (0.109091,-0.727273)
 --(0.109091,-0.145455)\dpicstop
\draw (1.309091,-0.581818) node{$=$};
\dpicdraw (2.4,-0.727273) rectangle (2.618182,-0.436364)\dpicstop
\draw (2.509091,-0.581818) node{$A^{T}$};
\dpicdraw (2.4,-0.145455) rectangle (2.618182,0.145455)\dpicstop
\draw (2.509091,0) node{$C$};
\dpicdraw (2.4,0.436364) rectangle (2.618182,0.727273)\dpicstop
\draw (2.509091,0.581818) node{$B$};
\dpicdraw (2,-0.727273)
 ..controls (2,-0.867854) and (2.113964,-0.981818)
 ..(2.254545,-0.981818)
 ..controls (2.395127,-0.981818) and (2.509091,-0.867854)
 ..(2.509091,-0.727273)\dpicstop
\dpicdraw (2.509091,-0.436364)
 --(2.509091,-0.145455)\dpicstop
\dpicdraw (2.509091,0.145455)
 --(2.509091,0.436364)\dpicstop
\dpicdraw (2.509091,0.727273)
 --(2.509091,1.018182)\dpicstop
\dpicdraw (2,-0.727273)
 --(2,1.018182)\dpicstop
\end{tikzpicture}